\documentclass[12pt,a4paper]{article}

\usepackage[utf8]{inputenc} 
\usepackage[T1]{fontenc}
\usepackage{url}
\usepackage{ifthen}
\usepackage{cite}
\usepackage[cmex10]{amsmath} % Use the [cmex10] option to ensure complicance
\usepackage{amssymb}
\usepackage{bm}
\usepackage{amsthm,amsfonts,ascmac}
\usepackage{cases}
\usepackage{mathtools}
\usepackage{color}

\theoremstyle{plain}

\newtheorem{thm}{Theorem}[section]
\newtheorem{dfn}[thm]{Definition}
\newtheorem{lem}[thm]{Lemma}

\newtheorem{fact}[thm]{Fact}
\newtheorem{ex}[thm]{Example}

\begin{document}
\title{The equivalence between correctability of\\deletions and insertions of separable states\\in quantum codes}

\author{
Taro Shibayama \thanks{
Department of Mathematics and Informatics,
Graduate School of Science,
Chiba University, Japan
Email: shibayama@kaijo.ed.jp
}
\and
Yingkai Ouyang \thanks{
Department of Electrical and Computer Engineering,
National University of Singapore, Singapore,
Email: yingkai@nus.edu.sg
}
}

\date{}
\maketitle

\begin{abstract}
In this paper, we prove the equivalence of inserting separable quantum states and deletions.
Hence any quantum code that corrects deletions automatically corrects separable insertions.
First, we describe the quantum insertion/deletion error using the Kraus operators.
Next, we develop an algebra for commuting Kraus operators corresponding to insertions and deletions.
Using this algebra, we prove the equivalence between quantum insertion codes and quantum deletion codes using the Knill-Laflamme conditions.
\end{abstract}

%%%%%%%%%%%%%%%%%%%%%%%%%%%%%%%%%%%%%%%%%%%%%%%%%%%%%%%%%%%%%%%%%%%%
%%%%%%%%%%%%%%%%%%%%%%%%%%%%%%%%%%%%%%%%%%%%%%%%%%%%%%%%%%%%%%%%%%%%
%%%%%%%%%%%%%%%%%%%%%%%%%%%%%%%%%%%%%%%%%%%%%%%%%%%%%%%%%%%%%%%%%%%%
\section{Introduction}

In quantum coding theory, erasures can be modeled using a partial trace where the traced qubits are known, but for deletions, we do not know what the traced qubits are.
We can also interpret deletions as erasures implemented by an adversary who hides information about which qubits were erased.
Hence, correcting deletions is harder than correcting erasures.
Similar to deletion errors, an insertion error occurs when a quantum state is inserted at unknown locations within a quantum code.

Quantum codes for a single deletion error were very recently studied by Nakayama and Hagiwara \cite{Nakayama20201,Hagiwara20202}.
A systematic construction of single deletion codes that encompasses these examples has been proposed \cite{Nakayama20202}, with more examples given by Shibayama \cite{9366159}.
Since it is clear that erasure and deletion errors are equivalent in permutation-invariant codes \cite{ouyang2021,shibayama2021},
Ouyang's permutation-invariant quantum codes \cite{Ouyang2014,Ouyang2017} are also quantum deletion codes.
Research on quantum codes correcting insertion errors on the other hand has only just begun, with the discovery that the four qubit deletion code can also correct a single insertion error \cite{ManabuHagiwara20212020XBL0191}.

Insertion/deletion error-correcting codes were first given by Levenshtein in 1966\cite{Levenshtein1966}, and it was shown that a code that can correct $t$ deletion errors can correct $t_1$ insertion errors and $t_2$ deletion errors if $t=t_1+t_2$.
This fact can be explained by using the Levenshtein distance\cite{971760}.
In classical coding theory, the development of codes for deletions and insertions is mature, and has a vast literature \cite{osti_10121528,9375003}. 
Indeed, classical insdel codes have received invigorated attention because of interesting applications such as for DNA storage \cite{TiloLeonid2013}, and racetrack memories \cite{8294216}.

This paper takes a novel perspective on quantum insertion/deletion errors by discussing them using the Kraus operator formalism \cite{Hellwig1969,Hellwig1970} for quantum channels.
This perspective allows us to use the Knill-Laflamme (KL) conditions for quantum error correction \cite{knill1997} which are traditionally written in terms of the Kraus operators of the noisy quantum channel.
Leveraging on the KL conditions, we prove the equivalence of the correctability between insertion of separable states and deletion errors for {\em any} quantum code.

The outline of the rest of this paper is as follows.
Section \ref{pre} introduces some notations and definitions, in particular the Kraus operators, the Knill-Laflamme conditions, and the quantum insertion/deletion channels.
The quantum insertion errors that we consider can introduce any separable state into unknown locations within a quantum code.
In this section, we also state the main theorem of this paper as Theorem \ref{thm25}.
In Section \ref{lems}, we state some lemmas about the calculation rules we require to prove the main theorem.
Section \ref{kraus} describes the Kraus form of insertion/deletion channels.
In Section \ref{prf}, we complete the proof of the main theorem.
Finally, we conclude this paper in Section \ref{conc}.

%%%%%%%%%%%%%%%%%%%%%%%%%%%%%%%%%%%%%%%%%%%%%
%%%%%%%%%%%%%%%%%%%%%%%%%%%%%%%%%%%%%%%%%%%%%
\section{Preliminaries}\label{pre}

Let $N$ be a positive integer and $[N]\coloneqq\{1,2,\dots ,N\}$.
For a square matrix $M$ over the complex field $\mathbb{C}$, we denote by ${\rm Tr}(M)$ the sum of the diagonal elements of $M$.
We fix $\mathcal{Q}\coloneqq\{0,1,\ldots,q-1\}$ for some integer $q\geq2$.
Let $|0\rangle, |1\rangle,\dots,|q-1\rangle$ be the standard orthonormal basis of $\mathbb{C}^q$, i.e.,
$|0\rangle\coloneqq(1,0,\dots,0)^\top,|1\rangle\coloneqq(0,1,0,\dots,0)^\top,\dots,|q-1\rangle\coloneqq(0,\dots,0,1)^\top$.
We denote $C_q\coloneqq\{|\psi\rangle\in\mathbb{C}^q\mid\langle\psi|\psi\rangle=1\}$ and let $|\Psi\rangle=|\psi_1\psi_2\cdots\psi_N\rangle\coloneqq|\psi_1\rangle\otimes|\psi_2\rangle\otimes\cdots\otimes|\psi_N\rangle\in (C_q)^{\otimes N}$.
Here, $\otimes$ is the tensor product operation and $\top$ is the transpose operation.
Let $\langle\bm{x}|\coloneqq|\bm{x}\rangle^\dag$ denote the conjugate transpose of $|\bm{x}\rangle$.
A positive semi-definite Hermitian matrix of trace $1$ is called a density matrix.
We denote by $S(\mathbb{C}^{q\otimes N})$ the set of all density matrices of order $q^N$. 
An element of $S(\mathbb{C}^{q\otimes N})$ is called an $N$-qudit quantum state.

\subsection{Kraus operators and Knill-Laflamme conditions}

Here, we review the Knill-Laflamme quantum error-correction criterion\cite{knill1997}.
A linear map $\Phi:S(\mathbb{C}^{q\otimes N})\rightarrow S(\mathbb{C}^{q\otimes N'})$ with positive integers $N,N'$ is a quantum channel, if and only if it is completely positive and trace-preserving.
For any quantum channel $\Phi:S(\mathbb{C}^{q\otimes N})\rightarrow S(\mathbb{C}^{q\otimes N'})$, there exist linear operators $A_i$ such that for every $\rho\in S(\mathbb{C}^{q\otimes N})$, $\Phi(\rho)=\sum_{i}A_i\rho A_i^\dag$ holds and $\sum_iA_i^\dag A_i$ is the identity operator on $S(\mathbb{C}^{q\otimes N})$.
The linear operators $A_i$ are known as Kraus operators of $\Phi$ and their representation is not unique.

Given a quantum channel $\mathcal{N}$, a subspace $\mathcal{C}$ of $\mathbb{C}^{q\otimes N}$ is a quantum code that corrects all errors introduced by $\mathcal{N}$, if and only if there exists a quantum channel $\mathcal{R}$ such that for every density matrix $\rho$ supported on $\mathcal{C}$, $\mathcal{R}(\mathcal{N}(\rho))=\rho$.
The necessary and sufficient conditions for quantum error-correction below were originally proved by Knill and Laflamme\cite{knill1997}.
\begin{fact}[Knill-Laflamme]\label{kl}
Let $\mathcal{C}$ be a $d$-dimensional subspace of complex Hilbert space $\mathbb{C}^{q\otimes N}$ with orthogonal basis vectors $|0_L\rangle,|1_L\rangle,\dots,|d-1_L\rangle$.
Let $\mathcal{N}$ be a quantum channel with Kraus operators $A_i$.
Suppose that for all $i,j$ there exist $g_{i,j}\in\mathbb{C}$ such that the following conditions hold.
\begin{itemize}
\item[1.] Orthogonality conditions:\\$\langle a_L|A_i^\dag A_j|b_L\rangle=0$ for all $a\neq b$.
\item[2.] Non-deformation conditions:\\$\langle a_L|A_i^\dag A_j|a_L\rangle=g_{i,j}$ for all $a=0,1,\dots,d-1$.
\end{itemize}
Then for every density matrix $\rho$ supported on $\mathcal{C}$, there exists a quantum channel $\mathcal{R}$ such that $\mathcal{R}(\mathcal{N}(\rho))=\rho$.
\end{fact}

\subsection{Quantum insertion and deletion channels}

Let $M=\sum_{\bm{x},\bm{y}\in\mathcal{Q}^N}m_{\bm{x},\bm{y}}|\bm{x}\rangle\langle\bm{y}|$ be a square matrix with $m_{\bm{x},\bm{y}}\in\mathbb{C}$.
For integers $p_1\in [N+1]$, $p_2\in [N]$, and a one qudit quantum state $\sigma\in S(\mathbb{C}^{q})$, define ${\rm In}_{p_1,\sigma}^N:S(\mathbb{C}^{q\otimes N})\rightarrow S(\mathbb{C}^{q\otimes (N+1)})$, and the partial trace ${\rm Tr}_{p_2}^N:S(\mathbb{C}^{q\otimes N})\rightarrow S(\mathbb{C}^{q\otimes (N-1)})$ respectively as 
\begin{align*}
{\rm In}_{p_1,\sigma}^N(M)\coloneqq&\sum_{\bm{x},\bm{y}\in\mathcal{Q}^N}m_{\bm{x},\bm{y}}|x_1\rangle\langle y_1|\otimes\dots\otimes|x_{p_1-1}\rangle\langle y_{p_1-1}|\\
&\otimes\sigma\otimes|x_{p_1}\rangle\langle y_{p_1}|\otimes\dots\otimes|x_N\rangle\langle y_N|,\\
{\rm Tr}_{p_2}^N(M)\coloneqq&\sum_{\bm{x},\bm{y}\in\mathcal{Q}^N}m_{\bm{x},\bm{y}}{\rm Tr}(|x_{p_2}\rangle\langle y_{p_2}|)|x_1\rangle\langle y_1|\otimes\cdots\otimes|x_{p_2-1}\rangle\langle y_{p_2-1}|\\
&\otimes|x_{p_2+1}\rangle\langle y_{p_2+1}|\otimes\cdots\otimes|x_N\rangle\langle y_N|.
\end{align*}

\begin{dfn}[$t$-insertion channel ${\rm Ins}_{t}^N$]\label{def22}
Let $t$ be a positive integer and let $\sigma=\sigma_1\otimes\sigma_2\otimes\dots\otimes\sigma_t\in S(\mathbb{C}^{q\otimes t})$, where $\sigma_i\in S(\mathbb{C}^q)$ is a one qudit quantum state for every $i\in [t]$.
For a set $P=\{p_1,p_2,\dots,p_t\}\subset [N+t]$ with $p_1<p_2<\dots<p_t$, define a map ${\rm Ins}_{P,\sigma}^N:S(\mathbb{C}^{q\otimes N})\rightarrow S(\mathbb{C}^{q\otimes(N+t)})$ as 
\begin{align*}
{\rm Ins}_{P,\sigma}^N(\rho)\coloneqq{\rm In}_{p_t,\sigma_t}^{N+t-1}\circ\dots\circ{\rm In}_{p_2,\sigma_2}^{N+1}\circ{\rm In}_{p_1,\sigma_1}^N(\rho),
\end{align*}
where $\rho\in S(\mathbb{C}^{q\otimes N})$ is a quantum state.
Here, the symbol $\circ$ denotes the composition of maps.
We call the map ${\rm Ins}_{P,\sigma}^N$ a $(t,P,\sigma)$-insertion error.
We define a $t$-insertion channel ${\rm Ins}_t^N$ as a convex combination of all $(t,P,\sigma)$-insertion errors,
 where $t$ is fixed and $|P|=t$, i.e.,
\begin{align*}
{\rm Ins}_t^N(\rho)\coloneqq\int_{\sigma
\in S(\mathbb{C}^{q\otimes t})
}\mu(\sigma)
\sum_{P:|P|=t}p_\sigma(P){\rm Ins}_{P,\sigma}^N(\rho)
d\sigma,
\end{align*}
where $\mu(\sigma)$ and $p_\sigma(P)$ are probability distributions.
Note that $\mu$ is a measure.
\end{dfn}

\begin{dfn}[$t$-deletion channel ${\rm Del}_t^N$]\label{def23}
Let $t<N$ be a positive integer.
For a set $P=\{p_1,p_2,\dots,p_t\}\subset [N]$ with $p_1<p_2<\dots<p_t$, define a map ${\rm Era}_P^N:S(\mathbb{C}^{q\otimes N})\rightarrow S(\mathbb{C}^{q\otimes(N-t)})$ as 
\begin{align*}
{\rm Era}_P^N(\rho)\coloneqq{\rm Tr}_{p_1}^{N-t+1}\circ\dots\circ{\rm Tr}_{p_{t-1}}^{N-1}\circ{\rm Tr}_{p_t}^{N}(\rho),
\end{align*}
where $\rho\in S(\mathbb{C}^{q\otimes N})$ is a quantum state.
We call the map ${\rm Era}_P^N$ a $(t,P)$-erasure error.
We define a $t$-deletion channel ${\rm Del}_t^N$ as a convex combination of all $(t,P)$-erasure errors, where $t$ is fixed and $|P|=t$, i.e.,
\begin{align*}
{\rm Del}_t^N(\rho)\coloneqq\sum_{P:|P|=t}p(P){\rm Era}_{P}^N(\rho),
\end{align*}
where $p(P)$ is a probability distribution.
\end{dfn}

Extending the above definitions, we take $0$-insertion and $0$-deletion channels to represent identity maps.

\begin{dfn}[$(t_1,t_2)$-insdel channel ${\rm InsDel}_{t_1,t_2}^{N}$]\label{def24}
Let $t_1, t_2$ be non-negative integers with $t_2<N$.
We define a $(t_1,t_2)$-insdel channel ${\rm InsDel}_{t_1,t_2}^{N}\!:\!S(\mathbb{C}^{q\otimes N})\rightarrow S(\mathbb{C}^{q\otimes (N+t_1-t_2)})$ as
\begin{align*}
{\rm InsDel}_{t_1,t_2}^{N}(\rho)\coloneqq{\rm Ins}_{t_1}^{N-t_2}\circ{\rm Del}_{t_2}^N(\rho),
\end{align*}
where $\rho\in S(\mathbb{C}^{q\otimes N})$ is a quantum state.
\end{dfn}

\subsection{Our main theorem}
A $d$-dimensional $(t_1,t_2)$-insdel code is a quantum code (a $d$-dimensional subspace of $\mathbb{C}^{q\otimes N}$) that can perfectly correct errors introduced by any $(t_1,t_2)$-insdel channel.
We denote by $\mathcal{C}\subset\mathbb{C}^{q\otimes N}$ the $(t_1,t_2)$-insdel quantum code spanned by the orthonormal logical codewords $|0_L\rangle,|1_L\rangle,\dots,|d-1_L\rangle$.
Our main theorem concerns $(t_1,t_2)$-insdel quantum codes.
In particular, we describe the equivalence between insertion and deletion error-correction capability in quantum codes.
The following Theorem \ref{thm25} is the main contribution of this paper.

\begin{thm}\label{thm25}
Let $t_1,t_2,s_1,s_2$ be non-negative integers where $t_1+t_2=s_1+s_2$.
Then, the $(t_1,t_2)$-insdel code is an $(s_1,s_2)$-insdel code.
\end{thm}

The remaining part of the paper is devoted to proving this. 

%%%%%%%%%%%%%%%%%%%%%%%%%%%%%%%%%%%%%%%%%%%%%
%%%%%%%%%%%%%%%%%%%%%%%%%%%%%%%%%%%%%%%%%%%%%
\section{Lemmas of tensor product calculation}\label{lems}
In this section, we introduce the rules of calculations necessary for the proof of our main theorem.

Let $n\geq0$, $t\geq1$ be integers and let $|\Psi\rangle=|\psi_1\psi_2\cdots\psi_t\rangle\in (C_q)^{\otimes t}$.
For a set $P=\{p_1,p_2\dots,p_t\}\subset [n+t]$ with $p_1<p_2<\dots<p_t$, we define $q^{n+t}$-by-$q^{n}$ matrix $I_{P,|\Psi\rangle}^{n}$ and $q^{n}$-by-$q^{n+t}$ matrix $D_{P,\langle\Psi|}^{n}$ as $I_{P,|\Psi\rangle}^{n}\coloneqq A_1\otimes A_2\otimes\dots\otimes A_{n+t}$ and $D_{P,\langle\Psi|}^{n}\coloneqq B_1\otimes B_2\otimes\dots\otimes B_{n+t}$, respectively, where
\begin{align*}
A_j\coloneqq
\begin{cases}
|\psi_i\rangle&j=p_i\in P,\\
\mathbb{I}_q&j\notin P,
\end{cases}~~
B_j\coloneqq
\begin{cases}
\langle\psi_i|&j=p_i\in P,\\
\mathbb{I}_q&j\notin p,
\end{cases}
\end{align*}
for $j\in[n+t]$.
Here, $\mathbb{I}_q$ denotes a size $q$ identity matrix.
Note that the superscript of these matrices represents the number of $\mathbb{I}_q$'s included as a factor of the tensor product.
When $t=1$, we simply denote $I_{P,|\Psi\rangle}^{n}$ and $D_{P,\langle\Psi|}^{n}$ as $I_{p_1,|\psi_1\rangle}^{n}$ and $D_{p_1,\langle\psi_1|}^{n}$, respectively.
It is clear from the definition that
\begin{align}
I_{P,|\Psi\rangle}^{n}={D_{P,\langle\Psi|}^{n}}^\dag,~~~~D_{P,\langle\Psi|}^{n}={I_{P,|\Psi\rangle}^{n}}^\dag.\label{eq1}
\end{align}

\begin{lem}\label{lem31} Let $n\geq0$, $t\geq1$ be integers.
Suppose that $P=\{p_1,p_2,\dots,p_t\}\subset [n+t]$ with $p_1<p_2<\cdots <p_t$ and $|\Psi\rangle=|\psi_1\psi_2\cdots\psi_t\rangle\in (C_q)^{\otimes t}$.
Then, 
\begin{itemize}
\item[1.] $I_{P,|\Psi\rangle}^{n}=I_{p_t,|\psi_t\rangle}^{n+t-1}I_{p_{t-1},|\psi_{t-1}\rangle}^{n+t-2}\cdots I_{p_{2},|\psi_{2}\rangle}^{n+1}I_{p_{1},|\psi_{1}\rangle}^{n}$.
\item[2.] $D_{P,\langle\Psi|}^{n}=D_{p_1,\langle\psi_1|}^{n}D_{p_2,\langle\psi_2|}^{n+1}\cdots D_{p_{t-1},\langle\psi_{t-1}|}^{n+t-2}D_{p_{t},\langle\psi_t|}^{n+t-1}$.
\end{itemize}
\end{lem}

\begin{lem}\label{lem32} Let $n$ be a non-negative integer, and let $p_1,p_2\in [n+2]$ with $p_1\leq p_2$ and $|\psi_1\rangle,|\psi_2\rangle\in C_q$.
Then,
\begin{itemize}
\item[1.] $I_{p_1,|\psi_1\rangle}^{n+1}I_{p_2,|\psi_2\rangle}^{n}=I_{p_2+1,|\psi_2\rangle}^{n+1}I_{p_1,|\psi_1\rangle}^{n}$.
\item[2.] $D_{p_2,\langle\psi_2|}^{n}D_{p_1,\langle\psi_1|}^{n+1}=D_{p_1,\langle\psi_1|}^{n}D_{p_2+1,\langle\psi_2|}^{n+1}$.
\end{itemize}
\end{lem}

Lemmas \ref{lem31} and \ref{lem32} can be easily shown by direct calculations as matrices.
For each lemma, after the first equation is shown, the second is easily shown from Equation (\ref{eq1}).

The following Lemmas \ref{lem33} and \ref{lem34} give commutation rules for the insertion Kraus operators $I$ and the deletion Kraus operators $D$. When $I$ and $D$ act on different qudits, they can be interchanged.
This is clear operationally, and we prove this algebraically.

\begin{lem}\label{lem33}
Let $n$ be a positive integer, and let $p_1,p_2\in [n+1]$ and $|\psi_1\rangle,|\psi_2\rangle\in C_q$.
Then, 
\begin{align*}
D_{p_2,\langle\psi_2|}^{n}I_{p_1,|\psi_1\rangle}^{n}=
\begin{cases}
I_{p_1,|\psi_1\rangle}^{n-1}D_{p_2-1,\langle\psi_2|}^{n-1}&p_1<p_2,\\
\langle\psi_2|\psi_1\rangle\mathbb{I}_{q^{n}}&p_1=p_2,\\
I_{p_1-1,|\psi_1\rangle}^{n-1}D_{p_2,\langle\psi_2|}^{n-1}&p_1>p_2.\\
\end{cases}
\end{align*}
\end{lem}

\begin{proof}
When $p_1<p_2$, simple calculations give
\begin{eqnarray*}
D_{p_2,\langle\psi_2|}^{n}I_{p_1,|\psi_1\rangle}^{n}
&=&(\mathbb{I}_{q^{p_1-1}}\otimes\mathbb{I}_q\otimes\mathbb{I}_{q^{p_2-p_1-1}}\otimes\langle\psi_2|\otimes\mathbb{I}_{q^{n-p_2+1}})\\
&&(\mathbb{I}_{q^{p_1-1}}\otimes|\psi_1\rangle\otimes\mathbb{I}_{q^{p_2-p_1-1}}\otimes\mathbb{I}_q\otimes\mathbb{I}_{q^{n-p_2+1}})\\
&=&\mathbb{I}_{q^{p_1-1}}\otimes|\psi_1\rangle\otimes\mathbb{I}_{q^{p_2-p_1-1}}\otimes\langle\psi_2|\otimes\mathbb{I}_{q^{n-p_2+1}},\\
I_{p_1,|\psi_1\rangle}^{n-1}D_{p_2-1,\langle\psi_2|}^{n-1}
&=&(\mathbb{I}_{q^{p_1-1}}\otimes|\psi_1\rangle\otimes\mathbb{I}_{q^{p_2-p_1-1}}\otimes\mathbb{I}_1\otimes\mathbb{I}_{q^{n-p_2+1}})\\
&&(\mathbb{I}_{q^{p_1-1}}\otimes\mathbb{I}_1\otimes\mathbb{I}_{q^{p_2-p_1-1}}\otimes\langle\psi_2|\otimes\mathbb{I}_{q^{n-p_2+1}})\\
&=&\mathbb{I}_{q^{p_1-1}}\otimes|\psi_1\rangle\otimes\mathbb{I}_{q^{p_2-p_1-1}}\otimes\langle\psi_2|\otimes\mathbb{I}_{q^{n-p_2+1}}.
\end{eqnarray*}
Hence, we obtain $D_{p_2,\langle\psi_2|}^{n}I_{p_1,|\psi_1\rangle}^{n}=I_{p_1,|\psi_1\rangle}^{n-1}D_{p_2-1,\langle\psi_2|}^{n-1}$.

The case $p_1>p_2$ is shown similarly, and the case $p_1=p_2$ trivially holds.
\end{proof}

\begin{lem}\label{lem34}
Let $n$ be a non-negative integer, and let $p_1,p_2\in [n+1]$ and $|\psi_1\rangle,|\psi_2\rangle\in C_q$.
Then, 
\begin{align*}
I_{p_1,|\psi_1\rangle}^{n}D_{p_2,\langle\psi_2|}^{n}=
\begin{cases}
D_{p_2+1,\langle\psi_2|}^{n+1}I_{p_1,|\psi_1\rangle}^{n+1}&p_1\leq p_2,\\
D_{p_2,\langle\psi_2|}^{n+1}I_{p_1+1,|\psi_1\rangle}^{n+1}&p_1\geq p_2.\\
\end{cases}
\end{align*}
\end{lem}
Lemma \ref{lem34} can be derived immediately from Lemma \ref{lem33}.

%%%%%%%%%%%%%%%%%%%%%%%%%%%%%%%%%%%%%%%%%%%%%
%%%%%%%%%%%%%%%%%%%%%%%%%%%%%%%%%%%%%%%%%%%%%
\section{Kraus operators for insdel errors}\label{kraus}

Here, we elucidate properties of the Kraus operators of insdel channels.
\begin{lem}\label{lem41}
For any quantum state $\rho\in S(\mathbb{C}^{q\otimes N})$, the state after inserting a separable state $\sigma\in S(\mathbb{C}^{q\otimes t})$ in the locations labeled by $P\subset [N+t]$ can be expressed as
\begin{align*}
{\rm Ins}_{P,\sigma}^N(\rho)=\sum_{\bm{a}\in\mathcal{Q}^{t}}p(\bm{a})I_{P,U|\bm{a}\rangle}^N\rho {I_{P,U|\bm{a}\rangle}^N}^\dag
\end{align*}
with some probability distribution $p(\bm{a})$ for $\bm{a}\in\mathcal{Q}^t$ and unitary matrix $U$.
\end{lem}

\begin{proof}
For any $n\geq1$, $p\in[n+1]$, $|\psi\rangle\in C_q$, and $\bm{x}\in\mathcal{Q}^n$,
\begin{eqnarray*}
I_{p,|\psi\rangle}^n|\bm{x}\rangle&=&(\mathbb{I}_{q^{p-1}}\otimes |\psi\rangle\otimes\mathbb{I}_{q^{n-p+1}})(|x_1\cdots x_{p-1}\rangle\otimes\mathbb{I}_1\otimes|x_p\cdots x_n\rangle)\nonumber\\
&=&|x_1\cdots x_{p-1}\psi x_p\cdots x_{n}\rangle
\end{eqnarray*}
holds.
Let $\tau=\sum_{i\in\mathcal{Q}}c_i|\psi_i\rangle\langle\psi_i|$ be the spectral decomposition of $\tau\in S(\mathbb{C}^q)$, where $c_i$ are probabilities and $\langle\psi_i|\psi_j\rangle=\delta_{i,j}$ for $i,j\in\mathcal{Q}$.
Here, $\delta_{i,j}$ is the Kronecker delta function.
Note that there exists a unitary matrix $U$ such that $|\psi_i\rangle=U|i\rangle$ for every $i\in\mathcal{Q}$.
For a quantum state $\rho=\sum_{\bm{x},\bm{y}\in\mathcal{Q}^n}m_{\bm{x},\bm{y}}|\bm{x}\rangle\langle \bm{y}|$, 
\begin{eqnarray*}
{\rm In}_{p,\tau}^n(\rho)&=&\sum_{i\in\mathcal{Q}}c_i\left({\textstyle\sum_{\bm{x},\bm{y}\in\mathcal{Q}^n}m_{\bm{x},\bm{y}}|\bm{x}_i\rangle\langle \bm{y}_i|}\right)\\
&=&\sum_{i\in\mathcal{Q}}c_iI_{p,|\psi_i\rangle}^n\left(\rule{0pt}{2.5ex}{\textstyle\sum_{\bm{x},\bm{y}\in\mathcal{Q}^n}m_{\bm{x},\bm{y}}|\bm{x}\rangle\langle \bm{y}|}\right){I_{p,|\psi_i\rangle}^n}^\dag\\
&=&\sum_{i\in\mathcal{Q}}c_iI_{p,U|i\rangle}^n\rho{I_{p,U|i\rangle}^n}^\dag
\end{eqnarray*}
holds, where $|\bm{x}_i\rangle=|x_1\cdots x_{p-1}\psi_i x_{p}\cdots x_n\rangle$ and $|\bm{y}_i\rangle=|y_1\cdots y_{p-1}\psi_i y_{p}\cdots y_n\rangle$.
Assume that $\sigma=\sigma_1\otimes\sigma_2\otimes\cdots\otimes\sigma_t$ and
$\sigma_k=\sum_{i\in\mathcal{Q}}c_i^k|\psi_i^k\rangle\langle\psi_i^k|$ and $|\psi_i^k\rangle=U_k|i\rangle$ for $k\in[t]$.
By Definition \ref{def22} and Lemma \ref{lem31}, we obtain
\begin{eqnarray*}
{\rm Ins}_{P,\sigma}^N(\rho)
&=&{\rm In}_{p_t,\sigma_t}^{N+t-1}\circ\dots\circ{\rm In}_{p_1,\sigma_1}^N(\rho)\\
&=&\sum_{i_t\in\mathcal{Q}}\dots\sum_{i_1\in\mathcal{Q}}c_{i_t}^t\cdots c_{i_1}^1I_{p_t,U_t|i_t\rangle}^{N+t-1}\cdots I_{p_1,U_1|i_1\rangle}^N \rho {I_{p_1,U_1|i_1\rangle}^N}^\dag\cdots{I_{p_t,U_t|i_t\rangle}^{N+t-1}}^\dag\\
&=&\sum_{\bm{a}\in\mathcal{Q}^{t}}p(\bm{a})I_{P,U|\bm{a}\rangle}^N\rho {I_{P,U|\bm{a}\rangle}^N}^\dag,
\end{eqnarray*}
where $p(\bm{a})=c_{a_1}^1\cdots c_{a_t}^t$ and $U=U_1\otimes\cdots\otimes U_t$.
\end{proof}

From Definition \ref{def22} and Lemma \ref{lem41}, we get the Kraus form for insertion channels, which is represented as
\begin{align}
{\rm Ins}_t^N(\rho)=\int_{U}
\sum_{P,\bm{a}}\mu_1(U,P,\bm{a})I_{P,U|\bm{a}\rangle}^N\rho{I_{P,U|\bm{a}\rangle}^N}^\dag
dU,\label{eq2}
\end{align}
where $\mu_1$ is a probability distribution.

\begin{lem}\label{lem42}
 For any quantum state $\rho\in S(\mathbb{C}^{q\otimes N})$, the state after deleting the qudits labeled by $P\subset [N]$ is
\begin{align*}
{\rm Era}_P^N(\rho)=\sum_{\bm{a}\in\mathcal{Q}^t}D_{P,\langle\bm{a}|}^{N-t}\rho {D_{P,\langle\bm{a}|}^{N-t}}^\dag.
\end{align*}
\end{lem}

\begin{proof}
For any $n\geq2$, $p\in[n]$, $a\in\mathcal{Q}$, and $\bm{x}\in\mathcal{Q}^n$, we have
\begin{eqnarray*}
D_{p,\langle a|}^{n-1}|\bm{x}\rangle
&=&(\mathbb{I}_{q^{p-1}}\otimes \langle a|\otimes\mathbb{I}_{q^{n-p}})(|x_1\cdots x_{p-1}\rangle\otimes|x_p\rangle\otimes|x_{p+1}\cdots x_n\rangle)\\
&=&\langle a|x_p\rangle|x_1\cdots x_{p-1}x_{p+1}\cdots x_{n}\rangle.
\end{eqnarray*}
For a quantum state $\rho=\sum_{\bm{x},\bm{y}\in\mathcal{Q}^n}m_{\bm{x},\bm{y}}|\bm{x}\rangle\langle \bm{y}|$,
\begin{eqnarray*}
{\rm Tr}_p^n(\rho)&=&\sum_{\bm{x},\bm{y}\in\mathcal{Q}^n}m_{\bm{x},\bm{y}}{\rm Tr}(|x_p\rangle\langle y_p|)|\bm{x}'\rangle\langle\bm{y}'|\\
&=&\sum_{\bm{x},\bm{y}\in\mathcal{Q}^n}m_{\bm{x},\bm{y}}
\left(\rule{0pt}{2.5ex}{\textstyle\sum_{a\in\mathcal{Q}}\langle a|x_p\rangle\langle y_p|a\rangle}\right)|\bm{x}'\rangle\langle\bm{y}'|\\
&=&\sum_{a\in\mathcal{Q}}\sum_{\bm{x},\bm{y}\in\mathcal{Q}^n}m_{\bm{x},\bm{y}}\langle a|x_p\rangle|\bm{x}'\rangle\langle\bm{y}'|\langle y_p|a\rangle\\
&=&\sum_{a\in\mathcal{Q}}D_{p,\langle a|}^{n-1}\left(\rule{0pt}{2.5ex}{\textstyle\sum_{\bm{x},\bm{y}\in\mathcal{Q}^n}m_{\bm{x},\bm{y}}|\bm{x}\rangle\langle \bm{y}|}\right){D_{p,\langle a|}^{n-1}}^\dag\\
&=&\sum_{a\in\mathcal{Q}}D_{p,\langle a|}^{n-1}\rho{D_{p,\langle a|}^{n-1}}^\dag
\end{eqnarray*}
holds, where $|\bm{x}'\rangle=|x_1\cdots x_{p-1}x_{p+1}\cdots x_n\rangle$ and $|\bm{y}'\rangle=|y_1\cdots y_{p-1}y_{p+1}\cdots y_n\rangle$.
Therefore, we have
\begin{eqnarray*}
{\rm Era}_{P}^N(\rho)
&=&{\rm Tr}_{p_1}^{N-t+1}\circ\dots\circ{\rm Tr}_{p_t}^{N}(\rho)\\
&=&\sum_{a_1\in\mathcal{Q}}\dots\sum_{a_t\in\mathcal{Q}}D_{p_1,\langle a_1|}^{N-t}\cdots D_{p_t,\langle a_t|}^{N-1}\rho{D_{p_t,\langle a_t|}^{N-1}}^\dag\cdots{D_{p_1,\langle a_1|}^{N-t}}^\dag\\
&=&\sum_{\bm{a}\in\mathcal{Q}^t}D_{P,\langle\bm{a}|}^{N-t}\rho {D_{P,\langle\bm{a}|}^{N-t}}^\dag
\end{eqnarray*}
by Definition \ref{def23} and Lemma \ref{lem31}.
\end{proof}

From Definition \ref{def23} and Lemma \ref{lem42}, we get the Kraus form for deletion channels, which is represented as
\begin{eqnarray}
{\rm Del}_t^N(\rho)&=&\sum_{P,\bm{a}}p(P)D_{P,\langle\bm{a}|}^{N-t}\rho{D_{P,\langle\bm{a}|}^{N-t}}^\dag\nonumber\\
&=&\int_{U}\sum_{P,\bm{a}}\mu_2(U,P,\bm{a})D_{P,\langle\bm{a}|U^\dag}^{N-t}\rho{D_{P,\langle\bm{a}|U^\dag}^{N-t}}^\dag dU.\label{eq3}
\end{eqnarray}
Note that by writing in integral form with a probability distribution $\mu_2$ as in (\ref{eq3}), the deletion channel can be regarded as having an infinite number of Kraus operators, just like the insertion channel.

Lemma \ref{lem43} below describes the intuitive result that deleting an inserted qudit leaves the original state unchanged.
We can see this by directly calculating the Kraus operator.
This lemma indicates that the operation of deleting after insertion is also included in the insdel error described in Definition \ref{def24}.
\begin{lem}\label{lem43}
Let $P=\{p\}\subset [N+1]$ and $\sigma\in S(\mathbb{C}^{q})$.
Then, for any quantum state $\rho\in S(\mathbb{C}^{q\otimes N})$,
\begin{align*}
{\rm Era}_P^{N+1}\circ{\rm Ins}_{P,\sigma}^N(\rho)=\rho.
\end{align*}
\end{lem}

For any quantum state $\rho\in S(\mathbb{C}^{q\otimes N})$, the state after insdel error described in Definition \ref{def24} can be calculated as
\begin{eqnarray}
\!\!\!{\rm InsDel}_{t_1,t_2}^{N}(\rho)
&=&\iint_{U,V}\sum_{P,Q,\bm{a},\bm{b}}\mu_{\bm{u}} I_{P,U|\bm{a}\rangle}^{N-t_2}D_{Q,\langle\bm{b}|V^\dag}^{N-t_2}\rho {D_{Q,\langle\bm{b}|V^\dag}^{N-t_2}}^{\!\dag} {I_{P,U|\bm{a}\rangle}^{N-t_2}}^{\!\dag} dUdV\label{eq4}
\end{eqnarray}
by Equations (\ref{eq2}) and (\ref{eq3}), where $\mu_{\bm{u}}$ is a non-negative value that depends on $\bm{u}=(U,V,P,Q,\bm{a},\bm{b})$.
We can easily calculate the matrix $I_{P,U|\bm{a}\rangle}^{N-t_2}D_{Q,\langle\bm{b}|V^\dag}^{N-t_2}$ such as in the example below.

\begin{ex}
Let $N=4$, $t_1=3$, $t_2=2$, $P=\{2,3,5\}$, $Q=\{1,3\}$, $U|\bm{a}\rangle=|\Psi\rangle=|\psi_1\psi_2\psi_3\rangle\in (C_q)^{\otimes 3}$, and $V|\bm{b}\rangle=|\Phi\rangle=|\phi_1\phi_2\rangle\in(C_q)^{\otimes 2}$.
Then, $I_{P,U|\bm{a}\rangle}^{2}D_{Q,\langle\bm{b}|V^\dag}^{2}$ is one of the Kraus operators corresponding to the action of inserting the second, third, and fifth components after deleting the first and third components as follows:
\begin{align*}
|x_1x_2x_3x_4\rangle\rightarrow|x_2x_4\rangle\rightarrow|x_2\psi_1\psi_2x_4\psi_3\rangle.
\end{align*}
%Since the insdel error is obtained by combining the deletion error and the insertion error, it can be calculated as follows.
%\begin{eqnarray*}
%%\lefteqn{
%{\rm InsDel}_{P,|\Psi\rangle,Q}^{4}(\rho)
%\!\!&=&\!\!{\rm Ins}_{P,|\Psi\rangle}^2\circ{\rm Del}_{Q}^4(\rho)\\
%\!\!&=&\!\!\!\sum_{\bm{a}\in\{0,1\}^{2}}I_{P,|\Psi\rangle}^{2}D_{Q,\langle\bm{a}|}^2\rho {D_{Q,\langle \bm{a}|}^2}^\dag {I_{P,|\Psi\rangle}^{2}}^\dag.
%\end{eqnarray*}
The matrices $I_{P,U|\bm{a}\rangle}^{2}$ and $D_{Q,\langle\bm{b}|V^\dag}^2$ can be expressed as
\begin{align*}
\begin{array}{rcccccccccccc}
I_{P,|\Psi\rangle}^{2}&\!=\!&\!\!\!\!&\!\!\!\!&\mathbb{I}_2&\!\!\otimes\!\!&|\psi_1\rangle&\!\!\otimes\!\!&|\psi_2\rangle&\!\!\otimes\!\!&\mathbb{I}_2&\!\!\otimes\!\!&|\psi_3\rangle,\\
D_{Q,\langle\Phi|}^2&\!=\!&\langle \phi_1|&\!\!\otimes\!\!&\mathbb{I}_2&\!\!\otimes\!\!&\langle \phi_2|&\!\!\!\!&\!\!\!\!&\!\!\otimes\!\!&\mathbb{I}_2&\!\!\!.\!&\\
\end{array}
\end{align*}
Therefore, we obtain
\begin{eqnarray*}
I_{P,|\Psi\rangle}^{2}D_{Q,\langle\Phi|}^2&=&\langle\phi_1|\,\otimes\,\mathbb{I}_2\,\otimes\,|\psi_1\rangle\,\otimes\,|\psi_2\rangle\,\otimes\,\langle\phi_2|\,\otimes\,\mathbb{I}_2\,\otimes\,|\psi_3\rangle.
\end{eqnarray*}
Note that $\bm{x}\bm{y}^\dag=\bm{x}\otimes\bm{y}^\dag=\bm{y}^\dag\otimes\bm{x}$ holds for any vectors $\bm{x},\bm{y}$.
\end{ex}

%%%%%%%%%%%%%%%%%%%%%%%%%%%%%%%%%%%%%%%%%%%%%
%%%%%%%%%%%%%%%%%%%%%%%%%%%%%%%%%%%%%%%%%%%%%
\section{Proof of the main theorem}\label{prf}
The aim of this section is to prove the main theorem. 

From Equation (\ref{eq4}) and Lemma \ref{lem31}, the Kraus operator $A_{\bm{u}}$ for $(t_1,t_2)$-insdel channel can be expressed as a product of $(t_1+t_2)$ block matrices
\begin{eqnarray}
A_{\bm{u}}&=&\sqrt{\mu_{\bm{u}}}\,I_{P,U|\bm{a}\rangle}^{N-t_2}D_{Q,\langle\bm{b}|V^\dag}^{N-t_2}\nonumber\\
&=&\sqrt{\mu_{\bm{u}}}\,I_{p_{t_1},|\psi_{t_1}\rangle}^{N-t_2+t_1-1}\cdots I_{p_{1},|\psi_{1}\rangle}^{N-t_2}
D_{q_{1},\langle\phi_1|}^{N-t_2}\cdots D_{q_{t_2},\langle \phi_{t_2}|}^{N-1},\label{eq*}
\end{eqnarray}
where $U|\bm{a}\rangle=|\psi_1\cdots\psi_{t_1}\rangle\in(C_q)^{\otimes t_1}$
 and $V|\bm{b}\rangle=|\phi_1\cdots\phi_{t_2}\rangle\in(C_q)^{\otimes t_2}$.
Therefore, the KL conditions for the $(t_1,t_2)$-insdel channel can be written as for all $i,j\in\{0,1,\dots,d-1\}$ and all $\bm{u}=(U,V,P,Q,\bm{a},\bm{b}),\bm{v}=(U',V',P',Q',\bm{a}',\bm{b}')$,
\begin{align}
\langle i_L|{D_{Q,\langle\bm{b}|V^\dag}^{N-t_2}}^{\!\dag}{I_{P,U|\bm{a}\rangle}^{N-t_2}}^{\!\dag} I_{P',U'|\bm{a}'\rangle}^{N-t_2}D_{Q',\langle\bm{b}'|V'^\dag}^{N-t_2}|j_L\rangle=\delta_{i,j}g_{\bm{u},\bm{v}}.\!\label{eq5}
\end{align}

The following two lemmas will help us establish the equivalence of insertion and deletions errors under the KL conditions.

\begin{lem}\label{lem51}
For $t_1\geq1$, any $(t_1,t_2)$-insdel quantum code is a $(t_1-1,t_2+1)$-insdel quantum code.
\end{lem}

\begin{proof}
From Equation (\ref{eq*}), we denote any two Kraus operators $B_{\bm{u}},B_{\bm{v}}$ for the $(t_1-1,t_2+1)$-insdel channel as
\begin{eqnarray*}
B_{\bm{u}}&=&I_{P,U|\bm{a}\rangle}^{N-(t_2+1)}D_{Q,\langle\bm{b}|V^\dag}^{N-(t_2+1)}\nonumber\\
&=&\underbrace{I_{p_{t_1-1},|\psi_{t_1-1}\rangle}^{N-t_2+t_1-3}\cdots I_{p_{1},|\psi_{1}\rangle}^{N-t_2-1}}_{(t_1-1) \textrm{ matrices}}
\underbrace{D_{q_{1},\langle\phi_1|}^{N-t_2-1}\cdots D_{q_{t_2+1},\langle\phi_{t_2+1}|}^{N-1}}_{(t_2+1) \textrm{ matrices}},\\
B_{\bm{v}}&=&I_{P',U'|\bm{a}'\rangle}^{N-(t_2+1)}D_{Q',\langle\bm{b}'|V'^\dag}^{N-(t_2+1)}\nonumber\\
&=&\underbrace{I_{p'_{t_1-1},|\psi'_{t_1-1}\rangle}^{N-t_2+t_1-3}\cdots I_{p'_{1},|\psi'_{1}\rangle}^{N-t_2-1}}_{(t_1-1) \textrm{ matrices}}
\underbrace{D_{q'_{1},\langle\phi'_1|}^{N-t_2-1}\cdots D_{q'_{t_2+1},\langle\phi'_{t_2+1}|}^{N-1}}_{(t_2+1) \textrm{ matrices}},
\end{eqnarray*}
where $U'|\bm{a}'\rangle=|\psi'_1\cdots\psi'_{t_1-1}\rangle$ and $V'|\bm{b}'\rangle=|\phi'_1\cdots\phi'_{t_2+1}\rangle$.
Note that when considering the KL condition, we can ignore the constant multiple of the Kraus operator.
By noting Equation (\ref{eq1}) and using Lemma \ref{lem34} repeatedly, we can calculate $B_{\bm{u}}^\dag B_{\bm{v}}$ using Fig \ref{fig1}.
Here, superscripts and subscripts are omitted to avoid confusion.
However, if we consider the superscripts, we can use Lemma \ref{lem34} in each matrix operation.
\begin{figure}[t]
\begin{align*}
\begin{array}{cccccccccccccccccccccl}
\!B_{\bm{u}}^\dag B_{\bm{v}}\!&\!=\!\!\!&\!\!\!D^\dag\!\!\!&\!\!\!\cdots\! &\!\!\!D^\dag\!\!\!&\!\!\!\colorbox[gray]{0.6}{\!\!$D^\dag$\!\!}\!\!\! &\!\!\!I^\dag\!\!\!&\!\!\!I^\dag\!\!\!&\!\cdots\! &\!\!\!I^\dag\!\!\!&\!\!\!I^\dag\!\!\!&\!\!\!I^\dag\!\!\! &\!I\!&\!I\!&\!I\!&\!\cdots\! &\!I\!&\!I\!&\!\!\colorbox[gray]{0.6}{$\!D\!$}\!\!&\!D\!&\!\cdots\!\! &\!D\!\\
&\!=\!\!\!&\!I\!&\!\!\!\cdots\! &\!I\!&\!\!\colorbox[gray]{0.6}{$I$}\!\! &\!D\!&\!D\!&\!\cdots\! &\!D\!&\!D\!&\!D\!&\!I\!&\!I\!&\!I\!&\!\cdots\! &\!I\!&\!I\!&\!\!\colorbox[gray]{0.6}{$\!D\!$}\!\!&\!D\!&\!\cdots\!\! &\!D\!\\
&\!=\!\!\!&\!I\!&\!\!\!\cdots\! &\!I\!&\!D\! &\!\!\colorbox[gray]{0.6}{$I$}\!\! &\!D\!&\!\cdots\! &\!D\!&\!D\!&\!D\!&\!I\!&\!I\!&\!I\!&\!\cdots\! &\!I\!&\!\!\colorbox[gray]{0.6}{$\!D\!$}\!\!&\!I\!&\!D\!&\!\cdots\!\! &\!D\!\\
&\!=\!\!\!&\!I\!&\!\!\!\cdots\! &\!I\!&\!D\! &\!D\!&\!\!\colorbox[gray]{0.6}{$I$}\!\! &\!\cdots\! &\!D\!&\!D\!&\!D\!&\!I\!&\!I\!&\!I\!&\!\cdots\! &\!\!\colorbox[gray]{0.6}{$\!D\!$}\!\!&\!I\!&\!I\!&\!D\!&\!\cdots\!\! &\!D\!\\
&\!=\!\!\!&\cdots&&&&&&&&\\
&\!=\!\!\!&\!I\!&\!\!\!\cdots\! &\!I\!&\!D\! &\!D\!&\!D\!&\!\cdots\! &\!\!\colorbox[gray]{0.6}{$I$}\!\!&\!D\!&\!D\!&\!I\!&\!I\!&\!\!\colorbox[gray]{0.6}{$\!D\!$}\!\!&\!\cdots\! &\!I\!&\!I\!&\!I\!&\!D\!&\!\cdots\!\! &\!D\!\\
&\!=\!\!\!&\!I\!&\!\!\!\cdots\! &\!I\!&\!D\! &\!D\!&\!D\!&\!\cdots\! &\!D\!&\!\!\colorbox[gray]{0.6}{$I$}\!\!&\!D\!&\!I\!&\!\!\colorbox[gray]{0.6}{$\!D\!$}\!\!&\!I\!&\!\cdots\! &\!I\!&\!I\!&\!I\!&\!D\!&\!\cdots\!\! &\!D\!\\
&\!=\!\!\!&\!I\!&\!\!\!\cdots\! &\!I\!&\!D\! &\!D\!&\!D\!&\!\cdots\! &\!D\!&\!D\!&\!\!\colorbox[gray]{0.6}{$I$}\!\!&\!\!\colorbox[gray]{0.6}{$\!D\!$}\!\!&\!I\!&\!I\!&\!\cdots\! &\!I\!&\!I\!&\!I\!&\!D\!&\!\cdots\!\! &\!D\!\\
&\!=\!\!\!&\!I\!&\!\!\!\cdots\! &\!I\!&\!D\! &\!D\!&\!D\!&\!\cdots\! &\!D\!&\!D\!&\!\!\colorbox[gray]{0.6}{$\!D\!$}\!\!&\!\!\colorbox[gray]{0.6}{$I$}\!\!&\!I\!&\!I\!&\!\cdots\! &\!I\!&\!I\!&\!I\!&\!D\!&\!\cdots\!\! &\!D\!\\
&\!=\!\!\!&\!\!\!D^\dag\!\!\!&\!\!\!\cdots\! &\!\!\!D^\dag\!\!\!&\!\!\!I^\dag\!\!\!&\!\!\!I^\dag\!\!\!&\!\!\!I^\dag\!\!\!&\!\cdots\! &\!\!\!I^\dag\!\!\!&\!\!\!I^\dag\!\!\!&\!\!\colorbox[gray]{0.6}{$\!I^\dag\!$}\!\!&\!\!\colorbox[gray]{0.6}{$I$}\!\!&\!I\!&\!I\!&\!\cdots\! &\!I\!&\!I\!&\!I\!&\!D\!&\!\cdots\!\! &\!D.\!\\
\end{array}
\end{align*}
\caption{Calculation of $B_{\bm{u}}^\dag B_{\bm{v}}$}
\label{fig1}
\end{figure}
Thus, $B_{\bm{u}}^\dag B_{\bm{v}}$ can be expressed as
\begin{eqnarray*}
B_{\bm{u}}^\dag B_{\bm{v}}
&=&\underbrace{D^\dag\cdots\, D^\dag}_{t_2+1}\underbrace{I^\dag\cdots\, I^\dag}_{t_1-1}
\underbrace{I\cdots\, I}_{t_1-1}\underbrace{D\cdots\, D}_{t_2+1}\\
&=&\underbrace{D^\dag\cdots\, D^\dag}_{t_2}\underbrace{I^\dag\cdots\, I^\dag}_{t_1}
\underbrace{I\cdots\, I}_{t_1}\underbrace{D\cdots\, D}_{t_2}.
\end{eqnarray*}
Furthermore, repeatedly applying Lemma \ref{lem32} gives $B_{\bm{u}}^\dag B_{\bm{v}}=A_{\bm{u}'}^\dag A_{\bm{v}'}$ for some Kraus operators $A_{\bm{u}'},A_{\bm{v}'}$ of the $(t_1,t_2)$-insdel channel.
From Equation (\ref{eq5}), we get 
$\langle i_L|B_{\bm{u}}^\dag B_{\bm{v}}|j_L\rangle=\delta_{i,j}g_{\bm{u}',\bm{v}'}$ 
for all $i,j\in\{0,1,\dots,d-1\}$ and all $\bm{u}',\bm{v}'$.
Since the pair $(\bm{u}',\bm{v}')$ is uniquely determined by $(\bm{u},\bm{v})$, the KL conditions for the $(t_1-1,t_2+1)$-insdel code hold for every $\bm{u},\bm{v}$.
Fact \ref{kl} implies that $\mathcal{C}$ is a $(t_1-1,t_2+1)$-insdel code.
\end{proof}

\begin{lem}\label{lem52}
For $t_2\geq1$, any $(t_1,t_2)$-insdel quantum code is a $(t_1+1,t_2-1)$-insdel quantum code.
\end{lem}

\begin{proof}
As in the proof of Lemma \ref{lem51}, denote any two Kraus operators $C_{\bm{u}},C_{\bm{v}}$ for the $(t_1+1,t_2-1)$-insdel channel as
\begin{eqnarray*}
C_{\bm{u}}&=&I_{P,U|\bm{a}\rangle}^{N-(t_2-1)}D_{Q,\langle\bm{b}|V^\dag}^{N-(t_2-1)}\\
&=&\underbrace{I_{p_{t_1+1},|\psi_{t_1+1}\rangle}^{N-t_2+t_1+1}\cdots I_{p_{1},|\psi_{1}\rangle}^{N-t_2+1}}_{(t_1+1) \textrm{ matrices}}
\underbrace{D_{q_{1},\langle\phi_1|}^{N-t_2+1}\cdots D_{q_{t_2-1},\langle\phi_{t_2-1}|}^{N-1}}_{(t_2-1) \textrm{ matrices}},\\
C_{\bm{v}}&=&I_{P',U'|\bm{a}'\rangle}^{N-(t_2-1)}D_{Q',\langle\bm{b}'|V'^\dag}^{N-(t_2-1)}\\
&=&\underbrace{I_{p'_{t_1+1},|\psi'_{t_1+1}\rangle}^{N-t_2+t_1+1}\cdots I_{p'_{1},|\psi'_{1}\rangle}^{N-t_2+1}}_{(t_1+1) \textrm{ matrices}}
\underbrace{D_{q'_{1},\langle\phi'_1|}^{N-t_2+1}\cdots D_{q'_{t_2-1},\langle\phi'_{t_2-1}|}^{N-1}}_{(t_2-1) \textrm{ matrices}}.\!\!\!
\end{eqnarray*}
This time using Lemma \ref{lem33} repeatedly, we can calculate $C_{\bm{u}}^\dag C_{\bm{v}}$ as Fig \ref{fig2}.
\begin{figure}[t]
\begin{align*}
\begin{array}{cccccccccccccccccccccl}
\!C_{\bm{u}}^\dag C_{\bm{v}}\!&\!=\!\!\!&\!\!\!D^\dag\!\!\!&\!\!\!\cdots\! &\!\!\!D^\dag\!\!\!&\!\!\!I^\dag\!\!\!&\!\!\!I^\dag\!\!\!&\!\!\!I^\dag\!\!\!&\!\cdots\! &\!\!\!I^\dag\!\!\!&\!\!\!I^\dag\!\!\!&\!\!\colorbox[gray]{0.6}{$\!I^\dag\!$}\!\!&\!\!\colorbox[gray]{0.6}{$I$}\!\!&\!I\!&\!I\!&\!\cdots\! &\!I\!&\!I\!&\!I\!&\!D\!&\!\cdots\!\! &\!D\!\\
&\!=\!\!\!&\!I\!&\!\!\!\cdots\! &\!I\!&\!D\! &\!D\!&\!D\!&\!\cdots\! &\!D\!&\!D\!&\!\!\colorbox[gray]{0.6}{$\!D\!$}\!\!&\!\!\colorbox[gray]{0.6}{$I$}\!\!&\!I\!&\!I\!&\!\cdots\! &\!I\!&\!I\!&\!I\!&\!D\!&\!\cdots\!\! &\!D\!\\
&\!=\!\!\!&\!I\!&\!\!\!\cdots\! &\!I\!&\!D\! &\!D\!&\!D\!&\!\cdots\! &\!D\!&\!D\!&\!\!\colorbox[gray]{0.6}{$I$}\!\!&\!\!\colorbox[gray]{0.6}{$\!D\!$}\!\!&\!I\!&\!I\!&\!\cdots\! &\!I\!&\!I\!&\!I\!&\!D\!&\!\cdots\!\! &\!D\!\\
&\!=\!\!\!&\!I\!&\!\!\!\cdots\! &\!I\!&\!D\! &\!D\!&\!D\!&\!\cdots\! &\!D\!&\!D\!&\!\!\colorbox[gray]{0.6}{$I$}\!\!&\!I\!&\!\!\colorbox[gray]{0.6}{$\!D\!$}\!\!&\!I\!&\!\cdots\! &\!I\!&\!I\!&\!I\!&\!D\!&\!\cdots\!\! &\!D\!\\
&\!=\!\!\!&\!I\!&\!\!\!\cdots\! &\!I\!&\!D\! &\!D\!&\!D\!&\!\cdots\! &\!D\!&\!D\!&\!\!\colorbox[gray]{0.6}{$I$}\!\!&\!I\!&\!I\!&\!\!\colorbox[gray]{0.6}{$\!D\!$}\!\!&\!\cdots\! &\!I\!&\!I\!&\!I\!&\!D\!&\!\cdots\!\! &\!D\!\\
&\!=\!\!\!&\cdots&&&&&&&&\\
&\!=\!\!\!&\!I\!&\!\!\!\cdots\! &\!I\!&\!D\! &\!D\!&\!D\!&\!\cdots\! &\!D\!&\!D\!&\!\!\colorbox[gray]{0.6}{$I$}\!\!&\!I\!&\!I\!&\!I\!&\!\cdots\! &\!\!\colorbox[gray]{0.6}{$\!D\!$}\!\!&\!I\!&\!I\!&\!D\!&\!\cdots\!\! &\!D\!\\
&\!=\!\!\!&\!I\!&\!\!\!\cdots\! &\!I\!&\!D\! &\!D\!&\!D\!&\!\cdots\! &\!D\!&\!D\!&\!\!\colorbox[gray]{0.6}{$I$}\!\!&\!I\!&\!I\!&\!I\!&\!\cdots\! &\!I\!&\!\!\colorbox[gray]{0.6}{$\!D\!$}\!\!&\!I\!&\!D\!&\!\cdots\!\! &\!D\!\\
&\!=\!\!\!&\!I\!&\!\!\!\cdots\! &\!I\!&\!D\! &\!D\!&\!D\!&\!\cdots\! &\!D\!&\!D\!&\!\!\colorbox[gray]{0.6}{$I$}\!\!&\!I\!&\!I\!&\!I\!&\!\cdots\! &\!I\!&\!I\!&\!\!\colorbox[gray]{0.6}{$\!D\!$}\!\!&\!D\!&\!\cdots\!\! &\!D\!\\
&\!=\!\!\!&\!I\!&\!\!\!\cdots\! &\!I\!&\!D\! &\!D\!&\!D\!&\!\cdots\! &\!D\!&\!\!\colorbox[gray]{0.6}{$I$}\!\!&\!D\!&\!I\!&\!I\!&\!I\!&\!\cdots\! &\!I\!&\!I\!&\!\!\colorbox[gray]{0.6}{$\!D\!$}\!\!&\!D\!&\!\cdots\!\! &\!D\!\\
&\!=\!\!\!&\!I\!&\!\!\!\cdots\! &\!I\!&\!D\! &\!D\!&\!D\!&\!\cdots\! &\!\!\colorbox[gray]{0.6}{$I$}\!\!&\!D\!&\!D\!&\!I\!&\!I\!&\!I\!&\!\cdots\! &\!I\!&\!I\!&\!\!\colorbox[gray]{0.6}{$\!D\!$}\!\!&\!D\!&\!\cdots\!\! &\!D\!\\
&\!=\!\!\!&\cdots&&&&&&&&\\
&\!=\!\!\!&\!I\!&\!\!\!\cdots\! &\!I\!&\!D\! &\!D\!&\!\!\colorbox[gray]{0.6}{$I$}\!\! &\!\cdots\! &\!D\!&\!D\!&\!D\!&\!I\!&\!I\!&\!I\!&\!\cdots\! &\!I\!&\!I\!&\!\!\colorbox[gray]{0.6}{$\!D\!$}\!\!&\!D\!&\!\cdots\!\! &\!D\!\\
&\!=\!\!\!&\!I\!&\!\!\!\cdots\! &\!I\!&\!D\! &\!\!\colorbox[gray]{0.6}{$I$}\!\! &\!D\!&\!\cdots\! &\!D\!&\!D\!&\!D\!&\!I\!&\!I\!&\!I\!&\!\cdots\! &\!I\!&\!I\!&\!\!\colorbox[gray]{0.6}{$\!D\!$}\!\!&\!D\!&\!\cdots\!\! &\!D\!\\
&\!=\!\!\!&\!I\!&\!\!\!\cdots\! &\!I\!&\!\!\colorbox[gray]{0.6}{$I$}\!\! &\!D\!&\!D\!&\!\cdots\! &\!D\!&\!D\!&\!D\!&\!I\!&\!I\!&\!I\!&\!\cdots\! &\!I\!&\!I\!&\!\!\colorbox[gray]{0.6}{$\!D\!$}\!\!&\!D\!&\!\cdots\!\! &\!D\!\\
&\!=\!\!\!&\!\!\!D^\dag\!\!\!&\!\!\!\cdots\! &\!\!\!D^\dag\!\!\!&\!\!\!\colorbox[gray]{0.6}{\!\!$D^\dag$\!\!}\!\!\! &\!\!\!I^\dag\!\!\!&\!\!\!I^\dag\!\!\!&\!\cdots\! &\!\!\!I^\dag\!\!\!&\!\!\!I^\dag\!\!\!&\!\!\!I^\dag\!\!\! &\!I\!&\!I\!&\!I\!&\!\cdots\! &\!I\!&\!I\!&\!\!\colorbox[gray]{0.6}{$\!D\!$}\!\!&\!D\!&\!\cdots\!\! &\!D.\!\\
\end{array}
\end{align*}
\caption{Calculation of $C_{\bm{u}}^\dag C_{\bm{v}}$}
\label{fig2}
\end{figure}
Note that, by Lemma \ref{lem33}, $DI=\langle\psi_2|\psi_1\rangle\mathbb{I}_{q^n}$ may occur in the middle of the calculation.
Thus, using $c_{\bm{u},\bm{v}}\in\mathbb{C}$ depending on $(\bm{u},\bm{v})$,
$C_{\bm{u}}^\dag C_{\bm{v}}$ can be expressed as
\begin{eqnarray*}
C_{\bm{u}}^\dag C_{\bm{v}}
&=&\underbrace{D^\dag\cdots\, D^\dag}_{t_2-1}\underbrace{I^\dag\cdots\, I^\dag}_{t_1+1}\underbrace{I\cdots\, I}_{t_1+1}\underbrace{D\cdots\, D}_{t_2-1}\\
&=&
\begin{cases}
c_{\bm{u},\bm{v}}\underbrace{D^\dag\cdots D^\dag}_{t_2-1}\underbrace{\,I^\dag\cdots \,I^\dag\,}_{t_1}\underbrace{\!\,\,I\,\cdots \,I\,\,\!}_{t_1}\underbrace{\,D\,\cdots \,D\,}_{t_2-1},\\
c_{\bm{u},\bm{v}}\underbrace{D^\dag\cdots D^\dag}_{t_2-1}\underbrace{\,I^\dag\cdots \,I^\dag\,}_{t_1-1}\underbrace{\!\,\,I\,\cdots \,I\,\,\!}_{t_1}\underbrace{\,D\,\cdots \,D\,}_{t_2},\\
\underbrace{D^\dag\cdots D^\dag}_{t_2}\underbrace{\,I^\dag\,\cdots \,I^\dag\,}_{t_1}\underbrace{\!\,\,I\,\,\cdots \,\,I\,\,\!}_{t_1}\underbrace{\,D\,\cdots \,D\,}_{t_2}.
\end{cases}
\end{eqnarray*}
By repeatedly applying Lemma \ref{lem32}, we obtain $C_{\bm{u}}^\dag C_{\bm{v}}=c_{\bm{u},\bm{v}}A_{\bm{u}'}^\dag A_{\bm{v}'}$ for some Kraus operators $A_{\bm{u}'},A_{\bm{v}'}$ of the $(t_1,t_2)$-insdel channel.
Note that for any non-negative integers $s_1\leq t_1$ and $s_2\leq t_2$, the $(t_1,t_2)$-insdel code is an $(s_1,s_2)$-insdel code.
From Equation (\ref{eq5}), we get 
$\langle i_L|C_{\bm{u}}^\dag C_{\bm{v}}|j_L\rangle=\delta_{i,j}c_{\bm{u},\bm{v}}g_{\bm{u}',\bm{v}'}$ 
for all $i,j\in\{0,1,\dots,d-1\}$ and all $\bm{u}',\bm{v}'$.
Since the pair $(\bm{u}',\bm{v}')$ is uniquely determined by $(\bm{u},\bm{v})$, the KL conditions for the $(t_1+1,t_2-1)$-insdel code hold for every $\bm{u},\bm{v}$.
From Fact \ref{kl}, it is shown that $\mathcal{C}$ is a $(t_1+1,t_2-1)$-insdel code.
\end{proof}

By Lemmas \ref{lem51} and \ref{lem52}, we have completed the proof of our main theorem, Theorem \ref{thm25}.

\section{Conclusion}\label{conc}
This paper provides proof of the equivalence of error-correction capability for quantum deletion and separable insertion errors.
Together constructions on permutation-invariant quantum codes \cite{Ouyang2014,Ouyang2017}, this implies the existence of quantum codes that correct separable insertions.
Broad range of permutation-invariant quantum codes \cite{Ouyang2014,Ouyang2016,OC2019} could be advantageous to use both in quantum storage \cite{Ouyang2020} and quantum metrology \cite{OSM2019} when both insertions and deletions occur.

\section*{Acknowledgment}
The authors thank Prof. Manabu Hagiwara for valuable discussions.
This paper is supported in part by KAKENHI 18H01435 and 21H03393.
Y.O. is supported in part by NUS startup grants (R-263-000-E32-133 and R-263-000-E32-731), and the National Research Foundation, Prime Minister's Office, Singapore and the Ministry of Education, Singapore under the Research Centres of Excellence programme.

\bibliography{bibtex}

\end{document}